\pdfoutput=1

\RequirePackage[l2tabu,orthodox]{nag}

\documentclass[11pt,letterpaper,reqno,oneside]{article}

\makeatletter
\renewcommand{\paragraph}{%
	\@startsection{paragraph}{4}%
	{\z@}{1.75ex \@plus 1ex \@minus .2ex}{-0.7em}%
	{\normalfont\normalsize\bfseries}%
}
\makeatother

\usepackage[T1]{fontenc}
\usepackage{lmodern}
\usepackage[utf8]{inputenc}

\usepackage[american,german,british]{babel}
\usepackage[activate={true,nocompatibility}]{microtype}

\usepackage{csquotes}

\usepackage[backend=biber,autolang=other,style=apa,maxcitenames=5,sortcites=false]{biblatex}
\DeclareLanguageMapping{british}{british-apa}
\DeclareLanguageMapping{american}{american-apa}
\usepackage{amsmath}
\usepackage{amssymb}
\usepackage{amsfonts}
\usepackage{amsthm}
\usepackage{mathtools}
\usepackage{stmaryrd}
\let\originalleft\left
\let\originalright\right
\renewcommand{\left}{\mathopen{}\mathclose\bgroup\originalleft}
\renewcommand{\right}{\aftergroup\egroup\originalright}

\usepackage{graphicx}
\usepackage{tikz,pgfplots}
\pgfplotsset{compat=1.10}
\usetikzlibrary{shapes.multipart}
\usetikzlibrary{decorations.markings}
\usepackage[title]{appendix}

\usepackage{datetime}
\newdateformat{datestyle}{ {\THEDAY} \monthname[\THEMONTH] \THEYEAR }

\usepackage{enumerate}
\usepackage{enumitem}
\setlist[enumerate,1]{label=(\arabic*)}
\setlist[itemize,1]{label=--}
\setlist[itemize,2]{label=--}
\setlist[itemize,3]{label=--}
\setlist[itemize,4]{label=--}

\usepackage{xcolor}
\usepackage{verbatim}
\usepackage{gensymb}
\usepackage{rotating}
\usepackage{subcaption}
\usepackage{floatpag}
\floatpagestyle{plain}
\usepackage{marvosym}
\usepackage{hyphenat}
\theoremstyle{definition}
\newtheorem{theorem}{Theorem}%[section]
\newtheorem{proposition}{Proposition}%[section]
\newtheorem{remark}{Remark}%[section]
\newtheorem{definition}{Definition}%[section]
\newtheoremstyle{named}
	{\topsep}					% ABOVESPACE
	{\topsep}					% BELOWSPACE
	{}							% BODYFONT
	{0pt}						% INDENT (empty value is the same as 0pt)
	{\bfseries}					% HEADFONT
	{}							% HEADPUNCT
	{5pt plus 1pt minus 1pt}	% HEADSPACE
	{\thmnote{#3}}				% CUSTOM-HEAD-SPEC
\theoremstyle{named}
\newtheorem{namedthm}{}

\usepackage{hyperref}
\hypersetup{pdfborder=0 0 0}

\usepackage[nameinlink]{cleveref}
\crefname{page}{p.}{pp.}
\crefname{equation}{equation}{equations}
\crefname{section}{section}{sections}
\crefname{subsection}{section}{sections}
\crefname{subsubsection}{section}{sections}
\crefname{appsec}{appendix}{appendices}
\crefname{supplsec}{supplemental appendix}{supplemental appendices}
\crefname{footnote}{footnote}{footnotes}
\crefname{figure}{figure}{figures}
\crefname{table}{table}{tables}
\crefname{theorem}{theorem}{theorems}
\crefname{proposition}{proposition}{propositions}
\crefname{lemma}{lemma}{lemmata}
\crefname{corollary}{corollary}{corollaries}
\crefname{remark}{remark}{remarks}
\crefname{observation}{observation}{observations}
\crefname{example}{example}{examples}
\crefname{fact}{fact}{facts}
\crefname{definition}{definition}{definitions}
\crefname{assumption}{assumption}{assumptions}
\crefname{exercise}{exercise}{exercises}
\crefname{notation}{notation}{notation}
\crefname{claim}{claim}{claims}
\crefname{conjecture}{conjecture}{conjectures}

\DeclareMathOperator*{\argmax}{arg\,max}

\DeclareMathOperator*{\co}{co}

\newcommand{\R}{\mathbf{R}}

\newcommand*{\xslant}[2][76]{%
	\begingroup
	\sbox0{#2}%
	\pgfmathsetlengthmacro\wdslant{\the\wd0 + cos(#1)*\the\wd0}%
	\leavevmode
	\hbox to \wdslant{\hss
		\tikz[
			baseline=(X.base),
			inner sep=0pt,
			transform canvas={xslant=cos(#1)},
		] \node (X) {\usebox0};%
		\hss
		\vrule width 0pt height\ht0 depth\dp0 %
	}%
	\endgroup
}
\makeatletter
\newcommand*{\xslantmath}{}
\def\xslantmath#1#{%
	\@xslantmath{#1}%
}
\newcommand*{\@xslantmath}[2]{%
	\ensuremath{%
		\mathpalette{\@@xslantmath{#1}}{#2}%
	}%
}
\newcommand*{\@@xslantmath}[3]{%
	\xslant#1{$#2#3\m@th$}%
}
\makeatother

\makeatletter
\def\namedlabel#1#2{\begingroup
	#2%
	\def\@currentlabel{#2}%
	\phantomsection\label{#1}\endgroup
}
\makeatother

\makeatletter
\let\save@mathaccent\mathaccent
\newcommand*\if@single[3]{%
	\setbox0\hbox{${\mathaccent"0362{#1}}^H$}%
	\setbox2\hbox{${\mathaccent"0362{\kern0pt#1}}^H$}%
	\ifdim\ht0=\ht2 #3\else #2\fi
	}
\newcommand*\rel@kern[1]{\kern#1\dimexpr\macc@kerna}
\newcommand*\widebar[1]{\@ifnextchar^{{\wide@bar{#1}{0}}}{\wide@bar{#1}{1}}}
\newcommand*\wide@bar[2]{\if@single{#1}{\wide@bar@{#1}{#2}{1}}{\wide@bar@{#1}{#2}{2}}}
\newcommand*\wide@bar@[3]{%
	\begingroup
	\def\mathaccent##1##2{%
	  \let\mathaccent\save@mathaccent
	  \if#32 \let\macc@nucleus\first@char \fi
	  \setbox\z@\hbox{$\macc@style{\macc@nucleus}_{}$}%
	  \setbox\tw@\hbox{$\macc@style{\macc@nucleus}{}_{}$}%
	  \dimen@\wd\tw@
	  \advance\dimen@-\wd\z@
	  \divide\dimen@ 3
	  \@tempdima\wd\tw@
	  \advance\@tempdima-\scriptspace
	  \divide\@tempdima 10
	  \advance\dimen@-\@tempdima
	  \ifdim\dimen@>\z@ \dimen@0pt\fi
	  \rel@kern{0.6}\kern-\dimen@
	  \if#31
	    \overline{\rel@kern{-0.6}\kern\dimen@\macc@nucleus\rel@kern{0.4}\kern\dimen@}%
	    \advance\dimen@0.4\dimexpr\macc@kerna
	    \let\final@kern#2%
	    \ifdim\dimen@<\z@ \let\final@kern1\fi
	    \if\final@kern1 \kern-\dimen@\fi
	  \else
	    \overline{\rel@kern{-0.6}\kern\dimen@#1}%
	  \fi
	}%
	\macc@depth\@ne
	\let\math@bgroup\@empty \let\math@egroup\macc@set@skewchar
	\mathsurround\z@ \frozen@everymath{\mathgroup\macc@group\relax}%
	\macc@set@skewchar\relax
	\let\mathaccentV\macc@nested@a
	\if#31
	  \macc@nested@a\relax111{#1}%
	\else
	  \def\gobble@till@marker##1\endmarker{}%
	  \futurelet\first@char\gobble@till@marker#1\endmarker
	  \ifcat\noexpand\first@char A\else
	    \def\first@char{}%
	  \fi
	  \macc@nested@a\relax111{\first@char}%
	\fi
	\endgroup
}
\makeatother

\title{\scshape Comparative risk attitude and the aggregation of single-crossing\thanks{Curello acknowledges support
from the German Research Foundation (DFG) through CRC TR 224 (Project B02).}}

\author{%
Gregorio Curello \\
Mannheim
\and
Ludvig Sinander \\
Oxford
\and
Mark Whitmeyer \\
Arizona State}

\date{2 December 2025}

\makeatletter
	\AtBeginDocument{ \hypersetup{
		pdftitle = {Comparative risk attitude and the aggregation of single-crossing},
		pdfauthor = {Gregorio Curello, Ludvig Sinander and Mark Whitmeyer}
		} }
\makeatother

\begin{document}

\maketitle

\begin{abstract}
In choice under risk, there is a standard notion of `less risk-averse than', due to \textcite{Yaari1969}. In the theory of comparative statics, the single-crossing property is satisfied by all weighted averages of a family of single-crossing functions if and only if the family satisfies a property called signed-ratio monotonicity \parencite{QuahStrulovici2012}. We establish a close link between `less risk-averse than' and signed-ratio monotonicity.
\end{abstract}

%%%%%%%%%%%%%%%%%%%%%%%%%%%%%%%%%%%
%%%%%%%%%%%%%%%%%%%%%%%%%%%%%%%%%%%
\section{Background: comparative risk attitude}
\label{sec:pratt}
%%%%%%%%%%%%%%%%%%%%%%%%%%%%%%%%%%%
%%%%%%%%%%%%%%%%%%%%%%%%%%%%%%%%%%%

For any non-empty finite set $X$, let $\Delta(X)$ be the set of all lotteries over $X$, i.e. all functions $p : X \to [0,1]$ such that $\sum_{x \in X} p(x) = 1$.

\begin{definition}[\cite{Yaari1969}]
	\label{definition:yaari}
	Let $X$ be a non-empty finite set, and fix functions $u,v : X \to \R$. We say that $u$ is \emph{less risk-averse than} $v$ if and only if for each $y \in X$ and each $p \in \Delta(X)$, $u(y) \geq \mathrel{(>)} \sum_{x \in X} u(x) p(x)$ implies $v(y) \geq \mathrel{(>)} \sum_{x \in X} v(x) p(x)$.
\end{definition}

Write `$\co A$' for the convex hull of a set $A \subseteq \R$, and `$\phi(Y)$' for the image of a function $\phi : Y \to \R$. Recall Pratt's (\citeyear{Pratt1964}) theorem:

\begin{namedthm}[Pratt's theorem.]
	\label{theorem:pratt}
	For a non-empty set $X$ and functions $u,v : X \to \R$, the following are equivalent:

	\begin{enumerate}[label=(\Alph*)]
	
		\item \label{item:pratt:lra} $u$ is less risk-averse than $v$.

		\item \label{item:pratt:trans} There exists an increasing convex function $\phi : \co(v(X)) \to \R$ that is strictly increasing on $v(X)$ and has $u(x) = \phi(v(x))$ for every $x \in X$.

		\item \label{item:pratt:curv} The following two properties hold:
		\begin{enumerate}[label=(\Roman*),topsep=0em]
		
			\item \label{item:pratt:curv:ordequiv} For any $x,y \in X$, $u(x) \geq \mathrel{(>)} u(y)$ implies $v(x) \geq \mathrel{(>)} v(y)$.

			\item \label{item:pratt:curv:compress} For any $x,y,z \in X$, if $u(x) < u(y) < u(z)$, then
			\begin{equation*}
				\frac{u(z)-u(y)}{u(y)-u(x)}
				\geq \frac{v(z)-v(y)}{v(y)-v(x)} .
			\end{equation*}
		
		\end{enumerate}
	
	\end{enumerate}
\end{namedthm}

\begin{remark}
	\label{remark:generalX}
	In the literature, comparative risk-aversion and its characterisation are almost only ever considered in the special case in which alternatives are monetary prizes: $X \subseteq \R$. However, by inspection, properties \ref{item:pratt:lra}--\ref{item:pratt:curv} are equally meaningful whatever the nature of the alternatives, and in fact they are equivalent even outside of the monetary-prizes case, as asserted above. See \textcite{CurelloSinanderWhitmeyer2025} for a proof.
\end{remark}

%%%%%%%%%%%%%%%%%%%%%%%%%%%%%%%%%%%
%%%%%%%%%%%%%%%%%%%%%%%%%%%%%%%%%%%
\section{Background: aggregation of single-crossing}
\label{sec:qs}
%%%%%%%%%%%%%%%%%%%%%%%%%%%%%%%%%%%
%%%%%%%%%%%%%%%%%%%%%%%%%%%%%%%%%%%

Abbreviate `partially ordered set' to `poset'.

\begin{definition}
	\label{definition:sc}
	Given a poset $(\Theta,\mathord{\lesssim})$, a function $\phi : \Theta \to \R$ is \emph{single-crossing} if and only if for any $\theta \lesssim \theta'$ in $\Theta$, $\phi(\theta) \geq \mathrel{(>)} 0$ implies $\phi(\theta') \geq \mathrel{(>)} 0$.
\end{definition}

This single-crossing property plays a central role in the theory of comparative statics. In particular, for a decision-maker with constraint set $X \subseteq \R$ and payoff function $U : X \times \Theta \to \R$, comparative-statics conclusions may be drawn about her choices $\theta \mapsto \argmax_{x \in X} U(x,\theta)$ if $U$ has \emph{single-crossing differences,} meaning precisely that for all $x<y$ in $X$, the function $U(y,\cdot) - U(x,\cdot)$ is single-crossing \parencite[see][]{MilgromShannon1994}.

We next consider \emph{families} of functions $\Theta \to \R$, i.e. subsets of $\R^\Theta$.

\begin{definition}[\cite{QuahStrulovici2012}]
	\label{definition:srm}
	Given a poset $(\Theta,\mathord{\lesssim})$, a family $\Phi$ of functions $\Theta \to \R$ satisfies \emph{signed-ratio monotonicity} if and only if for any $\phi,\psi \in \Phi$ and any $\theta \lesssim \theta'$ in $\Theta$, $\phi(\theta) < 0 < \psi(\theta)$ implies $-\phi(\theta)\psi(\theta') \geq -\phi(\theta')\psi(\theta)$.
\end{definition}

Obviously if the family $\Phi$ contains only single-crossing functions, then $\psi(\theta)>0$ implies $\psi(\theta')>0$, in which case the inequality $-\phi(\theta)\psi(\theta') \geq -\phi(\theta')\psi(\theta)$ may equivalently be written as $-\phi(\theta)/\psi(\theta) \geq -\phi(\theta')/\psi(\theta')$.

The following theorem is due to \textcite{QuahStrulovici2012}.%
	\footnote{Their statement of the theorem contains a typo (in particular, \ref{item:qs:sc_srm}\ref{item:qs:sc} is missing).}

\begin{namedthm}[Aggregation theorem.]
	\label{theorem:qs}
	For a non-empty poset $(\Theta,\mathord{\lesssim})$, a non-empty finite set $X$ and a function $f : X \times \Theta \to \R$, the following are equivalent:

	\begin{enumerate}[label=(\Alph*)]

		\item \label{item:qs:agg} For each $p \in \Delta(X)$, the map $\theta \mapsto \sum_{x \in X} f(x,\theta) p(x)$ is single-crossing.
	
		\item \label{item:qs:sc_srm} The following two properties hold:
		\begin{enumerate}[label=(\Roman*),topsep=0em]
		
			\item \label{item:qs:sc} For each $x \in X$, $f(x,\cdot)$ is single-crossing.

			\item \label{item:qs:srm} The family $\{ f(x,\cdot) : x \in X \}$ satisfies signed-ratio monotonicity.
		
		\end{enumerate}
	
	\end{enumerate}
\end{namedthm}

%%%%%%%%%%%%%%%%%%%%%%%%%%%%%%%%%%%
%%%%%%%%%%%%%%%%%%%%%%%%%%%%%%%%%%%
\section{The result}
\label{sec:equiv}
%%%%%%%%%%%%%%%%%%%%%%%%%%%%%%%%%%%
%%%%%%%%%%%%%%%%%%%%%%%%%%%%%%%%%%%

\begin{proposition}
	\label{proposition:equiv}
	For a non-empty poset $(\Theta,\mathord{\lesssim})$, a non-empty finite set $X$ and a function $U : X \times \Theta \to \R$, the following are equivalent:

	\begin{enumerate}[label=(\alph*)]

		\item \label{item:equiv:lra} For all $\theta \lesssim \theta'$ in $\Theta$, $U(\cdot,\theta)$ is less risk-averse than $U(\cdot,\theta')$.

		\item \label{item:equiv:qs} The following two properties hold:
		\begin{enumerate}[label=(\roman*),topsep=0em]
		
			\item \label{item:equiv:sc} For all $x,y \in X$, $U(y,\cdot) - U(x,\cdot)$ is single-crossing.

			\item \label{item:equiv:srm} For each $y \in X$, the family $\{ U(y,\cdot) - U(x,\cdot) : x \in X \}$ satisfies signed-ratio monotonicity.
		
		\end{enumerate}
	
	\end{enumerate}
\end{proposition}

\begin{proof}[Proof~1 (via the aggregation theorem)]
   \ref{item:equiv:lra} holds if and only if for all $\theta \lesssim \theta'$ in $\Theta$, each $y \in X$ and each $p \in \Delta(X)$, $U(y,\theta) - \sum_{x \in X} U(x,\theta) p(x) \geq \mathrel{(>)} 0$ implies $U(y,\theta') - \sum_{x \in X} U(x,\theta') p(x) \geq \mathrel{(>)} 0$. This is equivalent to: for each $y \in X$ and each $p \in \Delta(X)$, the map $\theta \mapsto \sum_{x \in X} [U(y,\theta)-U(x,\theta)] p(x)$ is single-crossing. By the \hyperref[theorem:qs]{aggregation theorem}, that is equivalent to \ref{item:equiv:qs}.
\end{proof}

\begin{proof}[Proof~2 (via Pratt's theorem)]
	\ref{item:equiv:qs}\ref{item:equiv:sc} holds if and only if for all $\theta \lesssim \theta'$ in $\Theta$ and all $x,y \in X$, $U(y,\theta) \geq \mathrel{(>)} U(x,\theta)$ implies $U(y,\theta') \geq \mathrel{(>)} U(x,\theta')$. This is equivalent to: for all $\theta \lesssim \theta'$ in $\Theta$, $u(\cdot) \coloneqq U(\cdot,\theta)$ and $v(\cdot) \coloneqq U(\cdot,\theta')$ satisfy property \ref{item:pratt:curv}\ref{item:pratt:curv:ordequiv} in \hyperref[theorem:pratt]{Pratt's theorem}.
    Furthermore, if \ref{item:equiv:qs}\ref{item:equiv:sc} is satisfied, then \ref{item:equiv:qs}\ref{item:equiv:srm} holds if and only if for all $\theta \lesssim \theta'$ in $\Theta$ and all $x,y,z \in X$, $U(y,\theta) - U(z,\theta) < 0 < U(y,\theta) - U(x,\theta)$ implies
    \begin{equation*}
        \frac{ U(z,\theta) - U(y,\theta) }
        { U(y,\theta) - U(x,\theta) }
        \geq \frac{ U(z,\theta') - U(y,\theta') }
        { U(y,\theta') - U(x,\theta') } .
    \end{equation*}
    In other words, if \ref{item:equiv:qs}\ref{item:equiv:sc} is satisfied, then \ref{item:equiv:qs}\ref{item:equiv:srm} holds if and only if for all $\theta \lesssim \theta'$ in $\Theta$, $u(\cdot) \coloneqq U(\cdot,\theta)$ and $v(\cdot) \coloneqq U(\cdot,\theta')$ satisfy property \ref{item:pratt:curv}\ref{item:pratt:curv:compress} in \hyperref[theorem:pratt]{Pratt's theorem}.
    Hence by \hyperref[theorem:pratt]{Pratt's theorem}, \ref{item:equiv:qs} is equivalent to \ref{item:equiv:lra}.
\end{proof}

%______________________________________________________________________________

%    ____  _ _     _ _                             _
%   | __ )(_) |__ | (_) ___   __ _ _ __ __ _ _ __ | |__  _   _
%   |  _ \| | '_ \| | |/ _ \ / _` | '__/ _` | '_ \| '_ \| | | |
%   | |_) | | |_) | | | (_) | (_| | | | (_| | |_) | | | | |_| |
%   |____/|_|_.__/|_|_|\___/ \__, |_|  \__,_| .__/|_| |_|\__, |
%                            |___/          |_|          |___/

% \pagebreak
\printbibliography[heading=bibintoc]

\end{document}